\documentclass[10pt,conference,letterpaper,doublecolumn]{ieeeconf}

\usepackage{bbm}
\usepackage{mathptmx}
\usepackage[dvips]{color}
\usepackage{epsf}
\usepackage{times}
\usepackage{booktabs}
\usepackage{setspace}
\usepackage{epsfig}
\usepackage{graphicx}
\usepackage{url}
\usepackage{amsmath}
\usepackage{amssymb}
\usepackage{amsxtra}
\usepackage{enumerate}

\usepackage{algorithmic}
\usepackage{algorithm}
\usepackage{cite}

\newtheorem{corollary}{\bf Corollary}
\newtheorem{theorem}{\bf Theorem}
\newtheorem{proposition}{\bf Proposition}

\newtheorem{definition}{\bf Definition}
\newtheorem{remark}{\bf Remark}

\IEEEoverridecommandlockouts

\title{\LARGE \bf SODEXO: A System Framework for Deployment and Exploitation of Deceptive Honeybots in Social Networks}

\author{
Quanyan Zhu, Andrew Clark, Radha Poovendran and Tamer Ba\c{s}ar%
\thanks{The research was partially supported by the AFOSR MURI Grant FA9550-10-1-0573,  and also by an NSA Grant through the Information Trust Institute at the University of Illinois.}%
\thanks{Q. Zhu and T. Ba\c{s}ar are with the Coordinated Science Laboratory and Department of Electrical and Computer Engineering, University of Illinois at Urbana-Champaign, Urbana, IL 61801 USA.
Email: {\tt\small \{zhu31, basar1\}@illinois.edu}}%
\thanks{Andrew Clark and Radha Poovendran are with the Department of Electrical Engineering, University of Washington,
Seattle, WA 98195 USA.
Email: {\tt\small \{awclark, rp3\}@u.washington.edu}}%
}

\begin{document}
\maketitle
\thispagestyle{empty}
\pagestyle{empty}

\begin{abstract}
As social networking sites such as Facebook and Twitter are becoming increasingly popular, a growing number of malicious attacks, such as phishing and malware, are exploiting them. 
Among these attacks, social botnets have  sophisticated infrastructure that leverages compromised users accounts, known as {\it bots}, to  automate the creation of new social networking
accounts for spamming and malware propagation. Traditional defense mechanisms are often passive and reactive to non-zero-day attacks. In this paper, we adopt a proactive approach for enhancing security in social networks by infiltrating botnets with honeybots. We propose an integrated system named SODEXO which can be interfaced with social networking sites for creating deceptive honeybots and leveraging them for gaining information from botnets. We establish a Stackelberg game framework to capture strategic interactions between honeybots and botnets, and use quantitative methods to understand the tradeoffs of honeybots for their deployment and exploitation in social  networks. We design a protection and alert system that integrates both microscopic and macroscopic models of honeybots and  optimally determines the security strategies for honeybots. We corroborate the proposed mechanism with extensive simulations and comparisons with passive defenses. 

\end{abstract}

\section{Introduction}
\label{sec:intro}
	Online social networks such as Facebook and Twitter are employed daily by hundreds of millions of users to communicate with acquaintances, follow news events, and exchange information.
	The growing popularity of OSNs has led to a corresponding increase in spam, phishing, and malware on social networking sites.
	The fact that a user is likely to click on a web link that appears in a friend's Facebook message or Twitter feed can be leveraged by attackers who compromise or impersonate that individual.

	An important class of malware attacks on social networks is social botnets~\cite{baltazar2009real,Keizer1998}.
	In a social botnet, an infected user's device and social networking account are both compromised by installed malware.
	The compromised account is then used to send spam messages to the user's contacts, containing links to websites with the malware executable.
	As a result, compromising a single well-connected user could lead to hundreds or thousands of additional users being targeted for spam, many of whom will also become members of the botnet and further propagate the malware.
	The most prominent example of a social botnet to date is Koobface, which at its peak had infected 600,000 hosts~\cite{baltazar2009real}.

	Current methods for mitigating malware, including social botnets, in social networks are primarily based on URL blacklisting.
	In this defense mechanism, links that are suspected to contain spam or malware are added to a centralized blacklist controlled by the owner of the social network.
	After a link has been blacklisted, the social networking site will no longer communicate with the IP address indicated by the link, even if a user clicks the link~\cite{twitter_faq}.

	While blacklisting can slow the propagation of malware, there remain several drawbacks to this approach.
	First, automated methods for blacklisting links often fail to detect spam and malware; one survey suggests that 73\% of malicious links go undetected and are not added to the blacklist~\cite{Thomas2010}.
	Second, automated blacklisting creates the risk of valid accounts and messages being classified as spam, degrading the user experience.
	Third, even for links that are correctly identified as pointing to malware, there is typically a large delay between when links are detected and blacklisted.
	 One study estimates this delay as 25 days on average, while at the same time most clicks on malware links occur within the first 48 hours of posting~\cite{grier2010spam}.

	A promising approach to defending against social botnets is through deception mechanisms.
	In a deceptive defense, the defender generates fake social network profiles that appear similar to real profiles and waits to receive a link to malware.
	The defender then follows the link to the malware site, downloads the malware executable, and runs it in a quarantined, sandbox environment.
	By posing as an infected node and interacting with the owner of the botnet, the defender  gathers links that are reported to the blacklist either before or shortly after they are posted, reducing the detection time and increasing the success rate.  
Currently, however, there is no
systematic approach to modeling social botnets and the effectiveness of deception, as well as designing an effective strategy for infiltrating the botnet and gathering information.
	
In this paper, we introduce an analytical framework for SOcial network DEception and eXploitation through hOneybots (SODEXO).  Our framework has two components, \emph{deployment} and \emph{exploitation}.  The deployment component models how decoy accounts are introduced into the online social network and gain access to the botnet.  The exploitation component characterizes the behavior of the decoys and the botnet owner after infiltration has occurred, enabling us to model the effect of the decoy on the botnet operation.

For the deployment component, we first develop a dynamical model describing the population of a social botnet over time.  We derive the steady-state equilibria of our model and prove the stability of the equilibria.  We then formulate the problem of selecting the optimal number of honeybots in order to maximize the information gathered from the botnet as a convex optimization problem.  Our results are extended to include networks with heterogeneous node degree.

We model the exploitation of the botnet by the honeybots as a Stackelberg game between the botmaster and the honeybots.  In the game, the botmaster allocates tasks, such as spam message delivery, among multiple bots based on their trustworthiness and capabilities.  The honeybots face a trade-off between obtaining more information by following the commands of the botmaster, and the impact of those commands on other network users.  We derive closed forms for the optimal strategies of both the botmaster and honeybots using \emph{Stackelberg equilibrium} as a solution concept.  We then incorporate the utility of the honeybot owner under the Stackelberg equilibrium in order to select an optimal deployment strategy.

The paper is organized as follows.  The related work is reviewed in Section \ref{sec:related}.  In Section \ref{sec:architecture}, we describe the architecture of our proposed framework for deceptive defense.  In Section \ref{sec:exploitation}, we model the exploitation phase of the botnet, in which the honeybot gathers the maximum possible information while avoiding detection by the botmaster.  In Section \ref{sec:deployment}, we model the deployment and population dynamics of the infected nodes and honeybots.  Section \ref{sec:PAS} describes the Protection and Alert System (PAS), which provides a unifying framework for controlling deployment and exploitation.  Section \ref{sec:simulation} presents our simulation results.  Section \ref{sec:conclusion} concludes the paper.  

%

\section{Related Work}
\label{sec:related}
%

Social botnets are becoming a serious threat for network users and managers, as they possess sophisticated infrastructure that leverages compromised users accounts, known as {\it bots}, to  automate the creation of new social networking
accounts for spamming and malware propagation \cite{Keizer1998}.
  In \cite{Lee2010}, a honeypot-based approach is used to
uncover social spammers in online social systems. It has been shown that
social honeypots can be used to identify social spammers with low false
positive rates, and that the harvested spam data contain signals
that are strongly correlated with observable profile features, such as friend information and posting patterns.  The goal of \cite{Lee2010}, however, is not to infiltrate the botnet, but to use honeypots to differentiate between real and spam online profiles.

In \cite{Thomas2010}, a zombie emulator is used to infiltrate the
Koobface botnet to discover the identities of fraudulent
and compromised social network accounts. The authors  arrived at the conclusion that ``to stem the threat of Koobface and the rise of social malware, social networks must advance their defenses beyond blacklists
and actively search for Koobface content, potentially using infiltration as a means of early detection." This insight coincides with our proactive approach for defending social networks using deceptive social honeybots.

Deception provides an effective approach for building proactively secure systems\cite{Caddell2004,DoD2004}. Considerable amount of work can be found using deception for enhancing cyber security.  In recent literature on intrusion detection systems, honeypots have been used to monitor suspicious intrusions \cite{Yuill2006,john2009studying}, and provide signatures of zero-day attacks \cite{Portokalidis2006}. 
In \cite{McQueen2009}, to enhance the security of  control systems in critical infrastructure, deception has been proposed to make the system more difficult for attackers to plan and execute successful attacks.  At present, however, there has been no analysis on the impact of deception on malware propagation in social networks.

In order to establish a formal method to evaluate the performance of deceptive social honeybots against botnets, we employ a game- and system- theoretic approach to model the strategic behaviors of botnets and the deployment and exploitations of honeybots. Such approaches have become pivotal for designing security mechanisms in a quantitative way \cite{ACMSurvey}.  In~\cite{khouzani2010maximum}, an optimal control approach to modeling the maximum impact of a malware attack on a communication network is presented. 
In \cite{ZhuCDC2011}, the authors have proposed an architecture for a collaborative intrusion detection network and have adopted a game-theoretic approach for designing a reciprocal incentive compatible resource allocation component of the system against free-rider and insider attacks.

\section{System Architecture}
\label{sec:architecture}

In this section, we introduce our honeybot-based defense system named SODEXO for protecting social networks against malicious attacks.  Fig. \ref{botdeception} illustrates the architecture of SODEXO.
\begin{figure}[t]
\centering
  \includegraphics[width=2.5in]{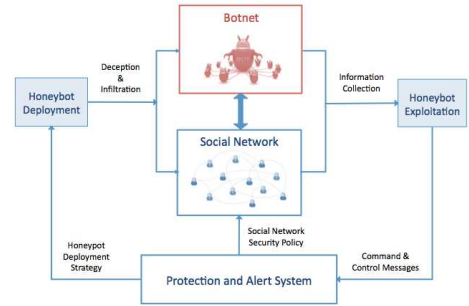}\\
  \caption{System architecture of honeybot deceptive mechanism in social networks}\label{botdeception}
\end{figure}
Our framework consists of two components,
 namely, honeybot deployment (HD) and honeybot exploitation (HE). HD deals with the distribution of honeybots within social networks and the deception mechanisms to infiltrate the botnet to learn and monitor the activities in botnets. HE aims to use the successfully infiltrated honeybots to collect as much  information as possible from the botnet.  The behaviors of the two blocks are coordinated by a Protection and Alert System (PAS), which uses the gathered information to generate real-time signatures and alerts for the social network (Fig. \ref{protection}).
 \begin{figure}[t]
\centering
  \includegraphics[width=2.5in]{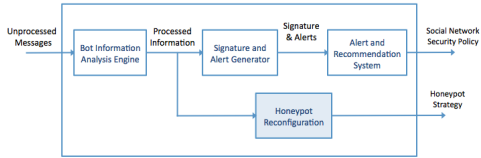}\\
  \caption{Architecture of the protection and alert system}\label{protection}
\end{figure}

The introduction of honeybots into a social network allows a proactive defense and monitoring of the social network against botnets. The SODEXO architecture bears its resemblance to feedback control systems. The HE component behaves as a security sensor of the social network; PAS can be seen as a controller which takes the ``measurements" from HE and yields a honeypot deployment strategy; and HD acts as an actuator that updates the honeypot policy designed by PAS.
In the following subsections, we discuss in detail each  component of SODEXO.
\subsection{Honeybot Deployment (HD)}
A honeypot is deployed by first creating an account on a social networking site.  The account profile is designed to imitate a real user, as in \cite{Lee2010}.  Once deployed, the honeypot sends a set of friend requests to a set of randomly chosen other users.  The honeypot continues sending friend requests to random users until the desired number of neighbors, denoted $d$, has been reached.  The honeypot monitors the message traffic of its neighbors, which may include personal messages, wall posts, or Twitter feeds, and follows any posted link.  If the link points to malware and has not been blacklisted, then the honeypot becomes a member of the social botnet and proceeds to the exploitation stage.

\subsection{Honeybot Exploitation (HE)}
The HE component of SODEXO takes advantage of the successfully infiltrated honeybots to gain as much information as possible from the botnet. The information is obtained in the form of command and control messages. The honeybots need to gain an appropriate level of trust from the bots and respond to the C\&C messages while minimizing harm to the legitimate social network users and avoiding legal liability. Honeybots  work collaboratively to achieve this goal. In the case where honeybots are commanded to send spam or malware to network users, they can send them to each other to remain active in the botnet. Depending on the sophistication of the botnet, honeybots can sometimes be detected using mechanisms described in \cite{Zou2010, Zou2006}. In this case, a higher growth rate of honeybot population will be needed to replace the detected honeybots. Hence, the performance of HE heavily depends on the effectiveness of HD, and in turn, HD should change its policy based on the sophistication of botnets and the amount of information learned in HE.

\subsection{Protection and Alert System (PAS)}

The major role of PAS is to provide security policies for HD  based on the information learned from HE.  Fig. \ref{protection} illustrates two major functions of PAS. The first step of PAS is to process the messages and logs gained from honeybots. Using data mining and machine learning techniques, it is possible that the structure of botnets can be inferred from network traffic information \cite{Gu_NDSS08_botsniffer} and  botnet C\&C channels in a local area network can be identified \cite{Gu_Scurity08_BotMiner}. These information can be used by the network administrator to detect the location of botmasters and remove them from the network.

The second important task of HD is to generate signatures for detecting malware and spam, which are then used to update the libraries of intrusion detection systems, blacklists of spam filters, and user alerts or recommendations.  The process of reconfiguration of IDSs and spam filters can be done either offline or real-time as in \cite{ZhuCDC2009} and \cite{GameSecZhu2011}. 

\begin{figure}[t]
\centering
  \includegraphics[width=2.5in]{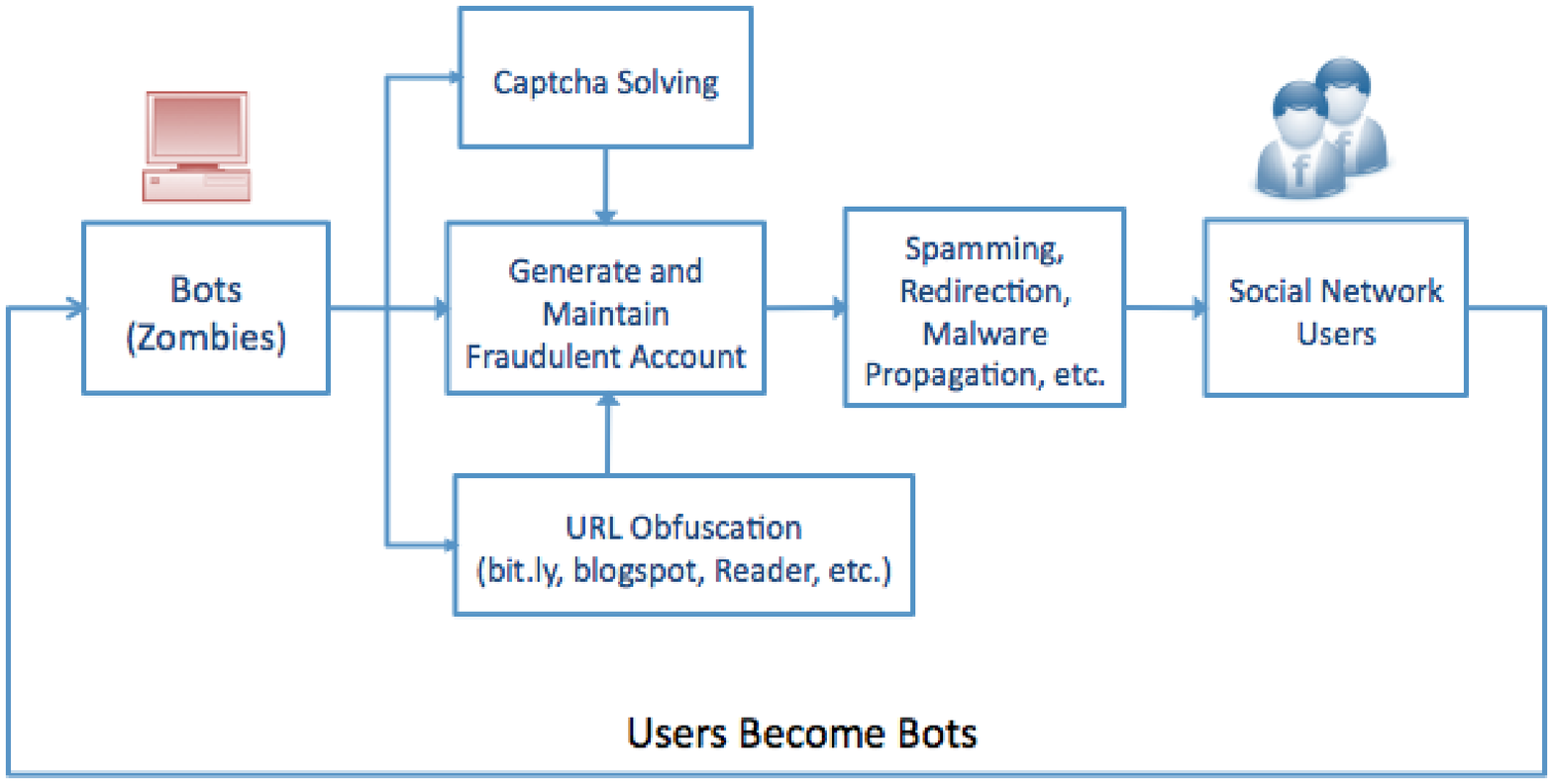}\\
  \caption{Illustration of the interactions between social networks and botnets}\label{Botnet-Socialnet}
\end{figure}

\subsection{Botnet Propagation Model}
 Fig. \ref{Botnet-Socialnet}  illustrates a mechanism used by botnets to infect social network users, which has been found in the Koobface botnet \cite{Keizer1998,Thomas2010}. The botnet maintains a fixed domain that bots or zombies regularly contact to report uptime statistics and request links for spamming activity. The bots  aim to obtain fresh user accounts and send malicious messages. The bot messages contain a malicious URL obfuscated by shortening services such as bit.ly or wrapped by an innocuous website including Google Reader and Blogger. Clicking on the URL of these messages eventually redirects to a spoofed Youtube or Facebook page that attempts to trick the victim into installing malware masquerading as a Flash update. Unsuspecting users  become infected by clicking on these messages.  Infected users are recruited to spam their own social network friends, leading to a wide propagation of malware within social network users.

 Once a user has been compromised, it makes frequent attempts to connect with one or more command control (C\&C) bots to retrieve commands from the botnet for further actions.  These commands are usually issued from another compromised computer for the purpose of concealing the botmaster's real identity~\cite{Zou2010}, leading to a hierarchical botnet architecture.  Fig. \ref{botnet2} illustrates the structure of a typical botnet, where a single botmaster sends messages to two C\&C bots and then they send to bots.

\section{System Model for Honeybot Exploitation}\label{sec:exploitation}




In this section, we introduce a system model for hierarchical botnets and employ a Stackelberg game framework to model the interactions between the botnet and infiltrating honeybots.

\subsection{Theoretical Framework}

Consider a botmaster $B$ that sends requests to a set of  C\&C bots $\mathcal{M}=\{1, 2, \cdots, m\}$ with $m=|\mathcal{M}|$. Each C\&C bot $i\in\mathcal{M}$ sends commands to a set of   compromised bot nodes $\mathcal{N}_i$ with $n_i=|\mathcal{N}_i|$. We assume that the botnet is a three-level tree architecture and, without loss of generality, we can assume $\cap_{i\in\mathcal{M}}\mathcal{N}_i =\emptyset$ since a single bot controlled by multiple C\&C bots can be modeled using multiple duplicate bots. Let $H$ be a honeybot that communicates with node $i\in\mathcal{M}$, i.e., $H \in\mathcal{N}_i$. We assume that all honeybots work together as a team, and hence one honeybot node $H$ under one C\& C subtree can conceptually represent a group of collaborative honeypots who have succeeded  in infiltrating the same botnet.

We let $p_{ij}\in\mathbb{R}_{+}$ be the number of messages or commands (in bytes) per second sent from C\&C bot $i$ to bot node $j\in\mathcal{N}_i$. Likewise, $p_{ji}$ denotes the number of response messages per second to C\&C node $i\in\mathcal{M}$ from node $j\in\mathcal{N}_i$.

 Each C\&C node $i$ maintains a trust value $T_{ij}\in [0,1]$ associated with a bot or honeybot node $j\in\mathcal{N}_i$. The trust values indicate the quality of response and performance of bot nodes. The trust values also inherently model the detection mechanisms in botnets, which have been discussed in \cite{Zou2010, Zou2010a}. For botnets with such mechanisms, low trust values indicate the inefficiency of a bot or a high likelihood of being a honeybot. For those without such mechanisms, we can take $T_{ij} = 1$, for all $j\in\mathcal{N}_i$, i.e., equivalently seeing all bots are all equally trusted.

One C\&C bot needs to send commands to a large population of bot nodes. Hence, the goal of  C\&C bot $i\in\mathcal{M}$ is to allocate its communication resources $\mathbf{p}_i:=[p_{ij}, j\in\mathcal{N}_i]$ to maximize the utility of its subtree network $U_i: \mathbb{R}^{n_i}_+ \rightarrow \mathbb{R} $, which is the sum of utilities obtained from each bot $j$, i.e.,
\begin{equation}\label{Ui}
U_i (\mathbf{p}_i) = \sum_{j\in\mathcal{N}_i} U_{ij},
\end{equation}
where $U_{ij}:\mathbb{R}_+ \rightarrow \mathbb{R}$ is the individual utility of C\&C bot $i$ from bot $j\in\mathcal{N}_i$, which is chosen to be
\begin{equation}\label{Uij}
U_{ij}(p_{ij}):=   T_{ij} p_{ji} \ln (\alpha_i p_{ij}+1).
\end{equation}
The choice of logarithmic function in (\ref{Uij}) indicates that the marginal utility of the C\&C bot diminishes as the number of messages increases. It captures the fact that the bots have limited resources to respond to commands, and a larger volume of commands can overwhelm the bots, which leads to diminishing marginal utility of node $i$. $\alpha_i\in\mathbb{R}_{++}$ is a positive system parameter that determines marginal utility.  

The utility of C\&C bot $i$ is also proportional to the number of messages or responses per second from bot $j$, indicated by $p_{ji} \in\mathbb{R}_+$. The number of response messages from bot $j$ indicates the level of activity of a bot. We can see that when $p_{ji} = 0$ or $T_{ij}=0$ in (\ref{Uij}), then bot $i$ is believed to be either inactive or fake, and it is equivalently removed from the subtree of C\&C node $j$ in terms of the total utility (\ref{Ui}). Note that $T_{ij}$ in (\ref{Uij}) evaluates the quality of the responses while $p_{ji}$ evaluates the quantity. The product of $T_{ij}$ and $p_{ji}$ captures the fact that the botnet values highly active and trusted bots.

We consider the following C\&C bot optimization problem (BOP) of every node $i\in \mathcal{M}$:
\begin{eqnarray}
\textrm{(BOP)} \nonumber \max_{\mathbf{p}_i\in\mathbb{R}^{n_i}_+}& U_i:= \sum_{j\in\mathcal{N}_i} T_{ij} p_{ji} \ln (\alpha_i p_{ij}+1)\\
\label{constraint} \textrm{\ s.t.\ } & \sum_{j\in\mathcal{N}_i} c_{ij} p_{ij} \leq C_i.
\end{eqnarray}

The constraint (\ref{constraint}) in (BOP) is a capacity constraint on the communications using C\&C channel, where $C_i$ is the total capacity of the channel.  The cost $c_{ij} \in\mathbf{R}_{++}$ is the cost on sending commands to bots. The cost is also dependent on the size of messages from C\&C bot $i$ to its controlled bots. It has been found in \cite{Thomas2010} that Twitter has larger volume of spam messages than Facebook. This is due to the fact that Twitter messages are often shorter than facebook messages, and hence the cost for commanding bots spamming with Twitter messages is relatively less than the one for Facebook.

Let $\mathcal{F}_i:=\{\mathbf{p}_i\in\mathbb{R}^{n_i}_+ :  \sum_{j\in\mathcal{N}_i} c_{ij} p_{ij} \leq C_i\}$ be the feasible set of (BOP). We let $\mathcal{L}_i: \mathbb{R}^{n_i}_+ \times \mathbb{R} \rightarrow \mathbb{R}$ be the associated Lagrangian defined as follows:
\begin{equation}
\mathcal{L}_i(\mathbf{p}_i, \lambda_i) =  \sum_{j\in\mathcal{N}_i} T_{ij} p_{ji} \ln (\alpha_i p_{ij}+1) + \lambda_i  \left(\sum_{j\in\mathcal{N}_i} c_{ij} p_{ij} - C_i\right)
\end{equation}
Since the feasible set is nonempty and convex, and the objective function is convex in $\mathbf{p}_i$, it is clear that (BOP) is a convex program, and hence we can use the first-order optimality condition to characterize the optimal solution to (BOP):
\begin{equation}
\frac{\partial \mathcal{L}_i}{\partial p_{ij}} =\sum_{j\in\mathcal{N}_i} \frac{\alpha_iT_{ij}p_{ji}}{\alpha_ip_{ij}+1} +\lambda_i c_{ij} = 0,
\end{equation}
which leads to
\begin{equation}\label{pij}
p_{ij} = -\frac{T_{ij}}{\lambda_i c_{ij}} - \frac{1}{\alpha_i}.
\end{equation}
Due to the monotonicity of logarithmic functions in (\ref{Uij}), the optimal solution is found on the Pareto boundary of feasible set. Hence by letting $\sum_{j\in\mathcal{N}_i} c_{ij} p_{ij} = C_i$, we obtain Lagrangian multiplier $\lambda_i$ from (\ref{pij}) as follows.
\begin{equation}\label{lambda}
\lambda_i = -\frac{\sum_{j\in\mathcal{N}_i}T_{ij}p_{ji}}{C_i+\frac{1}{\alpha_i}\sum_{j\in\mathcal{N}_i}c_{ij}}.
\end{equation}
We make following assumptions before stating Theorem \ref{thmBOP}.
\begin{enumerate}
\item [ \textbf{(A1)} ] The product $T_{ij}p_{ji} \neq0$ for all $j\in\mathcal{N}_i, i\in\mathcal{M}$.
\end{enumerate}

Assumption (A1) states that all bots controlled by C\&C bot $i$ are both active and trusted. This assumption is valid because for a controlled bot $j$ that is either inactive ($p_{ij}=0$) or untrusted ($T_{ij}=0$) can be viewed as the one  excluded from the set $\mathcal{N}_i$. Hence Assumption (A0) is equivalent to the statement that $\mathcal{N}_i$ contains all active and trusted bots.

\begin{theorem}\label{thmBOP}
Under Assumption (A1), (BOP) admits a unique solution when $\alpha_i$  is sufficiently large.
 \begin{equation}\label{pijclosed}
 p_{ij} = \left(\frac{T_{ij}p_{ji}}{\sum_{j\in\mathcal{N}_i}T_{ij}p_{ji}}\right)\left( \frac{C_i + \frac{1}{\alpha_i}\sum_{j\in\mathcal{N}_i} c_{ij}}{c_{ij}} \right)-\frac{1}{\alpha_i}.
 \end{equation}
 \end{theorem}

\begin{proof}
Assumption (A1) ensures that  (BOP) is strictly convex in $p_{ij}$ for all $j\in\mathcal{N}_i$. Hence the result follows directly from (\ref{pij}) and  (\ref{lambda}).  Since $\alpha_i$ is a system parameter, we can choose $\alpha_i$ sufficiently large so that the solution obtained in (\ref{pijclosed}) is nonnegative.
\end{proof}


\subsection{Stackelberg Game}

In this section, we formulate a two-stage Stackelberg between honeybots and C\&C nodes. Honeybots behave as leaders who can learn the behaviors of the C\&C bots once they succeed in infiltrating the botnet and choose the optimal strategies to respond to the commands from C\&C bots.

The goal of honeypots is to collect as much information as possible from the botmaster. 
We consider the following game between honeypots and a C\&C bot. The honeypot node $H$ firsts chooses a response rate $p_{Hi}$ to the commands from C\&C bot $i$, and then C\&C bot $i$ observes the response and chooses an optimal rate to send information to honeybot $H$ according to (BOP). We make the following assumption on the real bots in the network.

\begin{enumerate}
\item[\textbf{(A2)}] The real bots are not strategically interacting with the C\&C bot $i$, i.e., they send messages to bot $i$ at a constant rate $p_{ij}, j\neq H, j\in\mathcal{N}_i.$
\end{enumerate}

The above assumption holds because bots are non-human driven, pre-programmed to perform the same routine logic and communications as
coordinated by the same botmaster \cite{Gu_Scurity08_BotMiner}. Under Assumption (A2), the strategic interactions exist only between honeybots and C\&C nodes.

The honeypot node $H$ has a certain cost when it responds to the botnet. This can be either because of the potential harm that it can cause on the system or due to the cost of implementing commands from the botmaster. We consider the following honeypot optimization problem (HOP), where node $H$ aims to maximize its utility function $U_H: \mathbb{R}_+\times \mathbb{R}\rightarrow \mathbb{R}_+$ as follows:
\begin{eqnarray}\label{UH}
\nonumber \textrm{(HOP)} \max_{p_{Hi} \in \mathcal{F}_H} & U_H (p_{iH}, p_{Hi}):=\ln (p_{iH} + \xi_H) - \beta^H_i p_{Hi},
\end{eqnarray}
where $\xi_H \in\mathbb{R}_{++}$ is a positive system parameter; $\beta^H_i$ is the cost of honeybot $H$ responding to the bot node $i$; $p_{Hi}$ is the message sending rate from honeybot node $H$ to C\&C bot $i$ and $p_{Hi}$ is the rate of  C\&C bot $i$ sending commands to $H$.

$\mathcal{F}_H$ denotes the feasible set of the honeypot node $H$. We let $\mathcal{F}_H= \{p_{Hi}, 0 \leq p_{Hi} \leq p_{Hi,\max}\}$, where $p_{Hi,\max} \in\mathbb{R}_{++}$ is a positive parameter that can be chosen to be sufficiently large.
The logarithmic part of the utility function (\ref{UH}) is used to model the property of diminishing returns of an information source. The value of receiving an additional piece of information from the C\&C bot decreases as the total number of messages received by the honeypot increases. 


The interactions between honeypot $H$ and C\&C node $i$ can be captured by the Stackelberg game model $\Xi_S:=\langle (i, H),  (U_i, U_H), (\mathcal{F}_i, \mathcal{F}_H)\rangle$, and Stackelberg equilibrium can be used as a solution concept to characterize the outcome of the game.

\begin{definition}[Stackelberg Equilibrium]
\label{def:SE}
Let $\pi_{iH}^*(\cdot) : \mathbb{R}_+^{n_i} \rightarrow \mathbb{R}_+$ be the unique best response of the C\&C bots to the response rate $p_{Hi}$ of the honeypots. An action profile $(\mathbf{p}_{i}^*, \  p_{Hi}^*) \in \mathcal{F}_i \times \mathcal{F}_H$ is a {\it Stackelberg equilibrium} if
$\mathbf{p}_{i}^*= \pi_{iH}(p^*_{Hi}), $
 and the following inequality holds
$
U_H (\pi_{iH}(p^*_{Hi}), p^*_{Hi}) \geq U_H (\pi_{iH}(p_{Hi}), p_{Hi}), \ \forall p_{Hi} \in \mathcal{F}_H.
$
\end{definition}

\begin{theorem}
Under Assumption (A1), the nonzero-sum continuous-kernel Stackelberg game $\Xi_S$ admits a Stackelberg equilibrium.
\end{theorem}
\begin{proof}
The utility function of C\&C bot $i$ is strictly convex for all $p_{Hi} \neq 0$ under Assumption (A1). Since $\mathcal{F}_H$ and $\mathcal{F}_B$ are  compact sets, by Corollary 4.4 of \cite{Basar99dynamicnoncooperative}, the game admits a Stackelberg equilibrium solution.
\end{proof}

Under Assumption (A1), the unique best response $\pi_{iH}(\cdot)$ can be obtained from (\ref{pijclosed}) for sufficiently large $\alpha_i$ as follows:
\begin{equation}\label{pijclosedresponse}
p_{iH} = \pi_{iH} (p_{Hi}) = C_H  \left(\frac{T_{iH}p_{Hi}}{T_{iH}p_{Hi} + I_{-H}}\right) \ -\frac{1}{\alpha_i},
 \end{equation}
 where
$I_{-H} =\sum_{j\neq H, j\in\mathcal{N}_i}T_{ij}p_{ji}$ is the number of responses from real bots weighted by their trust values
and
$C_H:= \frac{C_i + \frac{1}{\alpha_i}\sum_{j\in\mathcal{N}_i} c_{ij}}{c_{iH}}$.

Letting $\bar{\xi}_H= 1/\alpha_i + \xi_H$ and substituting (\ref{pijclosedresponse}) in (HOP),  we  arrive at the following optimization problem faced by the honeybot node $H$: 
\begin{eqnarray}
\label{HOPUtility}
\nonumber \max_{p_{Hi}\in \mathcal{F}_H} U_H (\pi_{iH}(p_{Hi}), p_{Hi}):= \ \ \ \ \ \ \ \ \ \ \ \ \ \ \ \ \ \ \ \ \ \ \ \  \\   \ln \left(C_H \left(\frac{T_{iH}p_{Hi}}{T_{iH}p_{Hi} + I_{-H}}\right)\right.
\left.+ \bar{\xi}_H \right) - \beta^H_i p_{Hi}.
\end{eqnarray}
\begin{theorem}\label{thmSE}
Under Assumptions (A1) and (A2), the Stackelberg equilibrium solution $(\mathbf{p}_{i}^*, \  p_{Hi}^*)$ of the game $\Xi_S$ is unique and can be found as follows:
\begin{eqnarray}\label{SEsolOrg}
\nonumber p_{Hi}^*&=&\frac{C_H I_{-H}}{2T_{iH}(C_H+ \xi_{H})} \left( \sqrt{1  + 4 \frac{T_{iH} (C_H + \bar{\xi}_H)}{I_{-H} C_H \beta^H_i}} -1\right)\\
& & +\ \frac{I_{-H} \xi_H}{T_{iH}(C_H+ \xi_{H})},
\end{eqnarray}
and $p^*_{iH} = \pi_{iH}(p_{Hi}^*)$ and $p^*_{ij} = \pi_{ij} (p_{ij})$ for $j\neq H, j\in\mathcal{N}_i$.
\end{theorem}
\begin{proof}
The problem described in (\ref{HOPUtility}) is a convex program with the utility function $U_H$ convex in $p_{Hi}$ and convex set $\mathcal{F}$. Hence  the first-order optimality condition yields
\begin{flushleft}
$C_H I_{-H} T_{iH}=$
\end{flushleft}
\begin{equation}\label{FOQuadEqn}
\beta^H_i (I_{-H} + p_{Hi} T_{iH}) (C_H p_{Hi} T_{iH} + (I_{-H} + p_{Hi} T_{iH}) \bar{\xi}_H),
\end{equation}
which is a quadratic equation to be  solved for $p_{Hi}$ and its nonnegative solution of (\ref{FOQuadEqn}) is given in (\ref{SEsolOrg}). Since $p_{Hi,\max}$ is chosen sufficiently large and (\ref{SEsolOrg}) is non-negative,  $p_{iH}^*$ is a feasible solution. The equilibrium solution for bot $i$ hence follows from (\ref{pijclosedresponse}).
\end{proof}

In order to provide insights into the solution obtained in (\ref{SEsolOrg}), we make the following assumptions based on common structures of the botnets.
\begin{enumerate}
\item [\textbf{(A3)}] The real bots controlled by C\&C bot $i$ have identical features, i.e.,   $c_{ij} = \bar{c}_i$, $p_{ij} = \bar{p}_i$ and $T_{ij}= \bar{T}_i$ for all $j\neq H, j\in\mathcal{N}_i$.
\item [\textbf{(A4)}] The size of the real bots controlled by C\&C bot $i$ is much larger than the size of honeybots.
\item [\textbf{(A5)}] We let $\bar{\xi}_H=0$.
\end{enumerate}

Assumption (A5) is valid due to the freedom of choosing parameter $\xi_H$ in (HOP). Without loss of generality, we can let $\xi_H= \frac{1}{\alpha}_i$ and hence $\bar{\xi}_H=0$.
Assumption (A3) holds if real bots controlled by C\&C bot $i$ are of the same type, for example, Windows non-expert Facebook users. These type of users are commonly the target of botnets. Under (A3), we can simplify the expressions in (\ref{SEsolOrg}) and obtain $C_H = \frac{C_i}{c_i} + \frac{n_i}{\alpha_i}$, $I_{-H} = n_i^B\bar{T}_i\bar{p}_i$.

 Assumption (A4) is built upon the fact that one C\&C node in botnets often controls thousands of bots and the size of honeybots are often comparably small due to their implementation costs \cite{Provos:2007}. Under (A2), we have $I_{-H} \gg p_{Hi} T_{iH}$, then (\ref{HOPUtility}) can be rewritten as
 \begin{equation}\label{HOPUtility}
\tilde{U}_H (\pi_{iH}(p_{Hi}), p_{Hi})= \ln \left(C_H \left(\frac{T_{iH}p_{Hi}}{ I_{-H}}\right) + \bar{\xi}_H \right) - \beta^H_i p_{Hi}.
\end{equation}

\begin{corollary}\label{cor2}
Under Assumptions (A1), (A2) and (A4), the Stackelberg equilibrium solution $(\mathbf{p}_{i}^*, \  p_{Hi}^*) $ of the game $\Xi_S$ is given by
\begin{equation}\label{SimplifiedSolution}
p_{Hi}^*= \left(\frac{1}{\beta^H_i} - \frac{I_{-H} \xi_H}{C_H T_{iH} + T_{iH} \xi_H}\right)^+,
\end{equation}
where $(\cdot)^+ = \max\{0, \cdot\}$;
$p^*_{iH} = \pi_{iH}(p_{Hi}^*)$ and $p^*_{ij} = \pi_{ij} (p_{ij})$ for $j\neq H, j\in\mathcal{N}_i.$
\end{corollary}
\begin{proof}
From (A4), we can rewrite (\ref{FOQuadEqn})  by replacing $I_{-H} + p_{Hi} T_{iH}$ with $I_{-H}$. Since all the terms in (\ref{SimplifiedSolution}) is bounded, we can let $p_{Hi,\max}$ be sufficiently large and arrive at (\ref{SimplifiedSolution}).  The result then follows from Theorem \ref{thmSE}.
\end{proof}


\begin{corollary}\label{cor3}
Let the size of real bots under C\&C be $n_i^B$ and the size of the honeybots represented by super node $H$ $n_i^H$. Note that $n_i= n_i^B+n_i^H.$ Under Assumptions (A1) - (A5),  the Stackelberg equilibrium of the game $\Xi_S$ is given by
\begin{equation}
p_{Hi}^*=  \frac{1}{\beta^H_i},\ \ \ p^*_{ij} = \pi_{ij} (p_{ij}),\end{equation}
for $j\neq H, j\in\mathcal{N}_i$, and the equilibrium solution of  C\&C node $i$ is composed of two terms given by
$p^*_{iH} = p^*_{iH, S} + p^*_{iH, N},$
 with the first term independent of $n_i^H$,
\begin{equation}
p^*_{iH,S} =\frac{T_{iH}}{T_{iH} + \beta^H_i n_i^B \bar{T}_i\bar{p}_i}  \left(\frac{C_i}{c_i}+\frac{n_i^B}{\alpha_i}\right) - \frac{1}{\alpha_i},
\end{equation}
and the second term dependent on $n_i^H$,
\begin{equation}
p^*_{iH,S} =\frac{n_i^H T_{iH}}{T_{iH} + \beta^H_i n_i^B \bar{T}_i\bar{p}_i}.
\end{equation}
\end{corollary}
\begin{proof}
The result immediately follows from Corollary \ref{cor2} using (A3) and (A5).
\end{proof}

\begin{remark}
From Corollary \ref{cor3},  we can see that under Assumption (A1), the equilibrium response strategy is inversely proportional to the unit cost $\beta^H_i$.
We can see that the number of command and control messages harvested from the botnet is affine in the number of successfully infiltrated honeybots. The growth rate of the number of messages is given by
\begin{equation}
\label{eq:r_iH}
r^*_{iH} := \frac{\partial p^*_{iH} }{\partial n_i^H} = \frac{T_{iH}}{\beta^H_in_i^B \bar{T}_i \bar{p}_i + T_{iH}}.
\end{equation}
The growth rate is dependent of the trust value $T_{iH}$. Honeybots can harvest more information from the botnet if they are more trusted. The growth rate is also dependent on the number of the real bots controlled by C\&C bot $i$. As $n_i^B \rightarrow \infty$, the growth rate $r^*_{iH}\rightarrow 0$, i.e.,  size of honeybots will not affect the number of messages received by the network.
\end{remark}
%

Trust values can change over time and can either modeled by a random process or by some assessment rules adopted by the attacker. We can separate this into different subsections of discussion. We can also consider a dynamic optimization problem as well by having belief/trust as the state. This can be done through using beta or Dirichlet distributions. 


\begin{figure}[t]
\centering
  \includegraphics[width=2.5in]{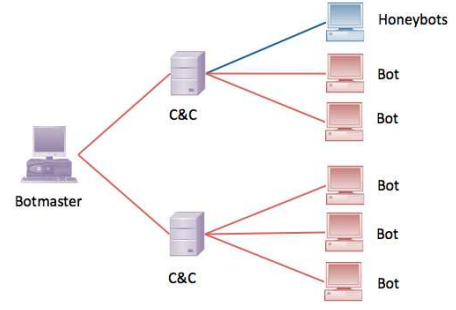}\\
  \caption{Illustration of a hierarchical social botnet}\label{botnet2}
\end{figure}

\section{Model of Honeybot Deployment and Botnet Growth}
\label{sec:deployment}
In what follows, a macroscopic model of the dynamics of the number of bots at time $t$, denoted $x_{1}(t)$, and the number of honeybots, denoted $x_{2}(t)$, is presented.  We then formulate an optimization problem for determining the number of honeypot nodes to introduce into the network.

  \subsection{Botnet and honeypot growth models}
  The bots are assumed to send spam messages, containing links to malware, with rate $r$.  Each message is sent to each of the $d$ neighbors of the bot, where $d$ is the average node degree.  Hence in each time interval $dt$, $rd \ dt$ spam messages are sent.  Since the number of valid nodes is $N-x_{1}(t)$, the number of messages reaching valid nodes is equal to $rd\frac{N-x_{1}(t)}{N} \ dt$.

The number of nodes that become bots depends on the behavior of the valid users and the number of links that have been blacklisted. Each user clicks on a spam link with probability $q$.  If the link has been blacklisted, then the user will be blocked from visiting the infected site; otherwise, the user's account is compromised and the device becomes part of the botnet.

To determine the probability that a link has been blacklisted, we assume that each bot is independently given a set of $k$ malicious links, out of $M$ links total.  The probability that a given link has been given to a specific honeybot is therefore $\frac{k}{M}$.  Hence the probability that a link has not been blacklisted is the probability that that link has not been given to any honeybot, which is equal to $\left(1 - \frac{k}{M}\right)^{x_{2}}$.  We assume that:
\begin{itemize}
\item[\textbf{(A6)}] The number of links given to each honeybot, $k$, satisfies $k \ll M$.
\end{itemize}
Under (A6), $\left(1 - \frac{k}{M}\right)^{x_{2}}$ can be approximated by $\left(1 - \frac{kx_{2}}{M}\right)$.

Finally, we assume that the infected devices are discovered and cleaned with rate $\mu_{1}$.  This leads to dynamics
\begin{equation}
\label{eq:bot_dynamics}
\dot{x}_{1}(t) = rdqx_{1}\left(1 - \frac{kx_{2}}{M}\right)\frac{N-x_{1}}{N} - \mu_{1}x_{1}.
\end{equation}

Honeybot nodes are inducted into the botnet in a similar fashion. We make the following assumptions regarding the honeybot population:
\begin{itemize}
\item[\textbf{(A7)}] The number of honeybots that are not part of the botnet, denoted $z$, is constant.
\item[\textbf{(A8)}] The number of honeybots is small compared to the total number of users, so that $\frac{z}{z+N} \approx \frac{z}{N}$.
\end{itemize}
Assumption (A7) can be guaranteed by creating new, uninfected honeybots when existing honeybots infiltrate the botnet.
 As with real users, honeybot nodes cannot follow links that have been blacklisted; however, unlike real users, honeybot nodes will attempt to follow any link with probability $1$.  The botmaster detects and removes honeybots with rate $\mu_{2}$.  The honeybot population is therefore defined by
\begin{equation}
\label{eq:honeypot_dynamics}
\dot{x}_{2}(t) = rdx_{1}\left(1 - \frac{kx_{2}}{M}\right)\frac{z}{N} - \mu_{2}x_{2}.
\end{equation}
\begin{proposition}
\label{prop:dynamic_equilibria}
The dynamics defined by (\ref{eq:bot_dynamics}) and (\ref{eq:honeypot_dynamics}) have two equilibria, given by $(x_{1}, x_{2}) = (0,0)$ and
\begin{equation}
\label{eq:dynamic_equilibria}
x_{1}^{\ast} = \frac{N\mu_{2}M(rdq-\mu_{1})}{rdq\mu_{2}M + rdkz\mu_{1}}, \quad x_{2}^{\ast} = \frac{\left(rd - \frac{\mu_{1}}{q}\right)z}{\frac{rdkz}{M} + \mu_{2}}.
\end{equation}
\end{proposition}

\begin{proof}
Equation (\ref{eq:bot_dynamics}) reaches equilibrium if $x_{1} = 0$ or if $rdq\left(1 - \frac{kx_{2}}{M}\right)\frac{N-x_{1}}{N} - \mu_{1} = 0$.  If $x_{1} = 0$, then (\ref{eq:honeypot_dynamics}) reaches equilibrium when $x_{2} = 0$.

If $rdq\left(1 - \frac{kx_{2}}{M}\right)\frac{N-x_{1}}{N} - \mu_{1} = 0$, solving for $x_{1}$ yields $x_{1} = N\left(1 - \frac{\mu_{1}}{rdq\left(1 - \frac{kx_{2}}{M}\right)}\right)$.  Substituting into (\ref{eq:honeypot_dynamics}) gives the equilibrium of (\ref{eq:dynamic_equilibria}).
\end{proof}

The quantity $rdq$ corresponds to the rate at which new nodes are inducted into the botnet, while $\mu_{1}$ is the rate at which nodes are cleaned and exit the botnet.  Thus if $rdq < \mu_{1}$, then the number of bots converges to zero, while $rdq > \mu_{1}$ implies that the number of bots converges to a nonzero steady-state value.

   Since network security policies are typically updated intermittently, while the dynamics of (\ref{eq:bot_dynamics}) and (\ref{eq:honeypot_dynamics}) converge rapidly, we base our subsequent analysis on the steady-state values of $x_{1}$ and $x_{2}$, and derive the optimal number of honeybots to introduce into the system in steady-state.
   In order to prove that this problem is well-defined, we first examine the stability properties of each equilibrium in the following theorem.

\begin{theorem}
\label{theorem:stability_properties}
If $\mu_{1} > rdq$, then $(x_{1},x_{2}) = (0,0)$ is asymptotically stable.  If $\mu_{1} < rdq$, then $(x_{1}^{\ast}, x_{2}^{\ast})$ is asymptotically stable in the limit as $M \rightarrow \infty$, $N \rightarrow \infty$.
\end{theorem}

\begin{proof}
An equilibrium point of a nonlinear dynamical system is asymptotically stable if its linearization is asymptotically stable at that point~\cite[Theorem 3.7]{khalil1992nonlinear}. 
The linearization of (\ref{eq:bot_dynamics}) and (\ref{eq:honeypot_dynamics}) around  $(0,0)$ is given by
\begin{displaymath}
A_{00} = \left(
\begin{array}{cc}
rdq - \mu_{1} & 0 \\
\frac{rdz}{N} & -\mu_{2}
\end{array}
\right).
\end{displaymath}
If $\mu_{1} > rdq$, then $-A_{00}$ is diagonally dominant, and hence has eigenvalues with positive real part [ref].  The eigenvalues of $A_{00}$ therefore have negative real part, implying that the linearization around $(0,0)$ is asymptotically stable.  The linearization $A_{x_{1}^{\ast}x_{2}^{\ast}}$ around $(x_{1}^{\ast}, x_{2}^{\ast})$ is given by
\begin{IEEEeqnarray*}{rCl}
A_{x_{1}^{\ast}x_{2}^{\ast}} &=& \left(
\begin{array}{cc}
a_{11} & a_{12} \\
a_{21} & a_{22}
\end{array}
\right), \textrm{ \ where \ \ }
\end{IEEEeqnarray*}
\begin{IEEEeqnarray*}{rCl}
a_{11} &=& \mu_{1} - \frac{rdq}{N}\left(\frac{\mu_{2}M + \frac{\mu_{1}kz}{q}}{\mu_{2}M + rdkz}\right)  \\
&& \cdot \left(-N + \frac{2\mu_{1}rdkz + \mu_{2}M}{rdq\mu_{2}M + rdkz\mu_{1}}\right),
\end{IEEEeqnarray*}
\begin{IEEEeqnarray*}{rCl}
 a_{12} &=& -\frac{rdqkN}{M}\left(1 - \frac{\mu_{1}rdkz + \mu_{2}M}{rdq\mu_{2}M + rdkz\mu_{1}}\right)\\
 && \cdot \left(\frac{\mu_{1}rdkz + \mu_{2}M}{rdq\mu_{2}M + rdkz\mu_{1}}\right),
 \end{IEEEeqnarray*}
\begin{IEEEeqnarray*}{rCl}
a_{21} &=& \frac{rdz}{N}\left(\frac{\mu_{2}M + \mu_{1}kz/q}{\mu_{2}M + rdkz}\right),
\end{IEEEeqnarray*}
\begin{IEEEeqnarray*}{rCl}
 a_{22} &=&-\frac{rdkz}{M}\left(1 - \frac{\mu_{1}rdkz + \mu_{2}M}{rdq\mu_{2}M + rdkz\mu_{1}}\right) - \mu_{2}.
\end{IEEEeqnarray*}
To prove that $A_{x_{1}^{\ast}x_{2}^{\ast}}$ has eigenvalues with negative real part, we examine $-A_{x_{1}^{\ast}x_{2}^{\ast}}$.  The second row is clearly diagonally dominant, since the diagonal element is positive and the off-diagonal element is negative.  The first row is diagonally dominant if
\begin{eqnarray}
\label{eq:diagonal_dominant}
\nonumber-\frac{rdqkN}{M}\left(1 - \frac{\mu_{1}rdkz + \mu_{2}M}{rdq\mu_{2}M + rdkz\mu_{1}}\right)\left(\frac{\mu_{1}rdkz + \mu_{2}M}{rdq\mu_{2}M + rdkz\mu_{1}}\right) \\
< \mu_{1} + rdq\left(\frac{\mu_{2}M + \frac{\mu_{1}kz}{q}}{\mu_{2}M + rdkz} - \frac{2\mu_{1}rdkz + \mu_{2}M}{rdq\mu_{2}M + rdkz\mu_{1}}\right).
\end{eqnarray}

In the limit as $M \rightarrow \infty$, the left-hand side of (\ref{eq:diagonal_dominant}) converges to zero while the right-hand side reduces to $\mu_{1} + rdq\left(1 - \frac{1}{rdqN}\right)$,
which is positive for $N$ sufficiently large.  Hence $-A_{x_{1}^{\ast}x_{2}^{\ast}}$ is diagonally dominant, and therefore has eigenvalues with positive real part, implying that $(x_{1}^{\ast}, x_{2}^{\ast})$ is a stable equilibrium point.
\end{proof}

\subsection{Computation of system parameters}
The parameter $\mu_{2}$ determines the rate at which honeypot nodes are discovered and removed by the botmaster, and hence can be calculated by observing the lifetime of deployed honeypots (see Section \ref{sec:PAS}).  Similarly, the number of received messages $p$ and the cost $\tau$ can be estimated by averaging over the set of deployed honeypots over time.  The fraction of malicious links $\frac{k}{M}$ that are given to a single bot or honeypot is estimated by using the assumption that links are distributed independently and uniformly at random by the botmaster, so that the probability that a given link has been received by a honeypot is $\left(1 - \frac{k}{M}\right)^{x_{2}}$.  This probability can be estimated by analyzing the set of malicious links received by new honeypots, which combined with knowledge of $x_{2}$ enables computation of $\frac{k}{M}$.  The rate at which spam messages are sent by bots, denoted $r$, is estimated by the number of instruction messages received by the honeypots.

The parameters $\mu_{1}$ and $q$, equal to the rate at which bots are removed from the botnet, and the fraction of malicious links that are followed by users, depend on user behavior.  These parameters can be estimated using existing data sets of user behavior~\cite{kanich2008spamalytics}.  Furthermore, to obtain an upper bound on the effectiveness of the botnet, the parameter $q$ can be set equal to $1$, implying that a valid user always clicks any link to the malware executable (the worst case).  The average node degree, $d$, is estimated based on existing analyses of the degree distribution of social networks~\cite{gjoka2011crawling}.

\subsection{Extension to heterogeneous networks}
Typical social networks follow a non-uniform degree distribution.  We present a model for the bot and honeypot population dynamics as follows.  Let $N_{d}$ denote the total number of users with degree $d$, and let $x_{1}^{d}(t)$ and $x_{2}^{d}(t)$ denote the number of bots and honeypots with degree $d$ at time $t$.  The average degree of the network is equal to $\overline{d}$.  We make the assumption that:
  \begin{itemize}
  \item[\textbf{(A9)}] The average degree of the infected users is equal to the average degree of the social network as a whole.
  \end{itemize}
  The total number of spam messages sent by bots in time interval $dt$ is equal to $r\overline{d}x_{1} \ dt$, each of which has not been blacklisted with probability $\left(1 - \frac{kx_{2}}{M}\right)$.  The probability that the recipient of a message has degree $d$ and has not been infected is equal to
\begin{displaymath}
Pr(\mbox{degree $d$, not infected}) = \frac{N_{d}-x_{1}^{d}}{N_{d}}\frac{d N_{d}}{\overline{d}N} = \frac{(N_{d}-x_{1}^{d})d}{\overline{d}N}.
\end{displaymath}
This implies that the dynamics of $x_{1}^{d}(t)$ are given by
\begin{equation}
\label{eq:bot_dynamics_heterogeneous}
\dot{x}_{1}^{d}(t) = \frac{rdq}{N}x_{1}\left(1 - \frac{kx_{2}}{M}\right)(N_{d} - x_{1}^{d}) - \mu_{1}x_{1}^{d},
\end{equation}
where $N_{d}$ is the number of user accounts of degree $d$.

Similarly, the probability that the recipient of a spam message is a honeypot node of degree $d$ that has not been infected yet is $\frac{z_{d}^{\ast}}{N}$, where $z_{d}$ is the number of honeybots of degree $d$ that have not joined the botnet, leading to dynamics of $x_{2}^{d}(t)$ described by
\begin{equation}
\label{eq:honeypot_dynamics_heterogeneous}
\dot{x}_{2}^{d}(t) = \frac{rd}{N}x_{1}\left(1 - \frac{kx_{2}}{M}\right)z_{d}^{\ast} - \mu_{2}x_{2}^{d}.
\end{equation}

\begin{proposition}
The dynamics (\ref{eq:bot_dynamics_heterogeneous}) and (\ref{eq:honeypot_dynamics_heterogeneous}) have equilibrium points at $x_{1}^{d} = x_{2}^{d} = 0$ for all $d$ and
\begin{IEEEeqnarray}{rCl}
\label{eq:x_1_heterogeneous}
x_{1}^{d\ast} &=& \left[rdq\left(1-\frac{ kr\overline{d}_{z}\left(1 - \frac{\mu_{1}}{r\overline{d}q}\right)z}{rkzd_{z} + \mu_{2}M} - \frac{\mu_{1}}{r\overline{d}q}\right) + \mu_{1}\right]^{-1} \\
\nonumber
&& \cdot \left(rdq\left(1 -\frac{ kr\overline{d}_{z}\left(1 - \frac{\mu_{1}}{r\overline{d}q}\right)z}{rkzd_{z} + \mu_{2}M} - \frac{\mu_{1}}{r\overline{d}q}\right)\right), \\
\label{eq:x_2_heterogeneous}
x_{2}^{d \ast} &=& \frac{rd}{\mu_{2}}\left(1 - \frac{kr\overline{d}_{Z}\left(1 - \frac{\mu_{1}}{r\overline{d}q}\right)z}{rkzd_{Z} + \mu_{2}M} - \frac{\mu_{1}}{r\overline{d}q}\right)z_{d}.
\end{IEEEeqnarray}
\end{proposition}

\begin{proof}
Summing (\ref{eq:bot_dynamics_heterogeneous}) over $d$ yields
\begin{displaymath}
\dot{x}_{1}(t) = \frac{x_{1}\overline{d}rq}{N}\left(1 - \frac{kx_{2}}{M}\right)(N - x_{1}) - \mu_{1}x_{1},
\end{displaymath}
which implies that in steady-state we have
\begin{equation}
\label{eq:SS_heterogeneous}
x_{1}^{\ast} = N\left(1 - \frac{\mu_{1}}{r\overline{d}q\left(1 - \frac{kx_{2}}{M}\right)}\right).
\end{equation}
Similarly, summing $\dot{x}_{2}^{d}(t)$ over $d$ results in
\begin{displaymath}
\dot{x}_{2}(t) = \frac{rx_{1}}{N}\left(1 - \frac{kx_{2}}{M}\right)z\overline{d}_{Z} - \mu_{2}x_{2},
\end{displaymath}
which combined with (\ref{eq:SS_heterogeneous}) gives
$x_{2}^{\ast} = \frac{r\overline{d}z\left(1 - \frac{\mu_{1}}{r\overline{d}q}\right)}{\frac{rkz\overline{d}_{z}}{M} + \mu_{2}}.$
The steady-state value (\ref{eq:x_2_heterogeneous}) can be obtained from (\ref{eq:honeypot_dynamics_heterogeneous}).  Similarly, $x_{1}^{d \ast}$ can be obtained from (\ref{eq:bot_dynamics_heterogeneous}).
\end{proof}



\section{Modeling of Protection and Alert System (PAS)}
\label{sec:PAS}
PAS is a coordination system that strategically deploys honeybots and designs security policies for social networks.
In this section, we focus on optimal reconfiguration of honeybots as illustrated in Fig. \ref{protection}.  We introduce a mathematical framework for finding  honeybot deployment strategies based on system models described in Sections \ref{sec:exploitation} and \ref{sec:deployment}.

\subsection{Relations between HD and HE}
We have adopted a divide-and-conquer approach in Sections \ref{sec:exploitation} and \ref{sec:deployment}, and have modeled the behavior of each system independently. However, the interdependencies between HD and HE are essential for PAS to make optimal security policies for the social network. The HE model in Section \ref{sec:exploitation} describes strategic operations of honeybots at a microscopic level while the HD model in Section \ref{sec:deployment} provides a macroscopic description of the population dynamics of bots and honeybots. These two models are interrelated  through their parameters together with the feedback control from PAS.

The interactions between bots and honeybots in the HE model occur on a time scale of seconds. The analysis of Stackelberg equilibrium in Section \ref{sec:exploitation} captures the steady-state equilibrium after a repeated or learning process of the game. Hence the equilibrium can be reached on a time scale of minutes. On the other hand, the population dynamics in HD model evolve on a larger time scale (for example, days). Hence, we can assume that the Stackelberg game has reached its equilibrium when the populations evolve at a macroscopic level. Decisions made at PAS are on a longer time scale (for example, weeks) because the processing of collected information, learning of bots and honeybots in social networks, and high-level decision on security policy in reality demand considerable amount of human resources for coordination and supervision.

\subsubsection{Trust Values and Detection Rate}
The trust values $T_{ij}$ used in HE model are related to the macroscopic detection and removal rate $\mu_2$ in HD model. As we have pointed out earlier, zero trust values are equivalent to the removal of honeybots from the botnet. Hence we can adopt a simple dynamic model to describe the change of $T_{ij}$ over a longer time period (say, days). We let $T_{ij}^0$ be the initial condition of the trust value. 
The evolution of $T_{ij}$ over the macroscopic time scale can be modeled using the following ODE:
\begin{equation}\label{TODE}
T_{ij}(t) = -\mu_2 T_{ij}(t), \ \  T_{ij}(t^0_{ij}) = T_{ij}^0. 
\end{equation}
Note that honeybots have different initial time $t^0_{ij}$. Hence from (\ref{TODE}), we obtain
\begin{equation}\label{TODE:sol}
T_{ij}(t) = e^{-\mu_2 (t-t^0_{ij})}, \ \ t \geq t^0_{ij},
\end{equation}
i.e., the trust values exponentially decay with respect to the removal rate. a threshold can be set on. From (\ref{TODE:sol}), we can obtain the mean life time of a honeybot is $1/\mu_2$. Macroscopic parameter $\mu_2$ can be estimated by the rate of change of working honeybots in the botnet, which can is known to the system, while $T_{ij}$ is a microscopic parameter and is often unknown directly to honeybots. With the ODE model in (\ref{TODE}), we can use $\mu_2$ to estimate $T_{ij}$.

\subsubsection{Honeybot and Bot Populations}
In Section \ref{sec:deployment}, the populations of bots and infiltrating honeybots are denoted by $x_1$ and $x_2$, respectively, whereas in Section \ref{sec:exploitation}, the bot size under C\&C bot is $n_i^B$. Under a hierarchical structure of botnet, illustrated in Fig. \ref{botnet2}, the total bot and honeybot populations $x_1, x_2$ are given by
\begin{equation}\label{x1Expr}
x_1 =\sum_{i=1}^m n_i^B, \ \ x_2=\sum_{i=1}^m n_i^H.
\end{equation}
If all C\&C bots are assumed to be identical, i.e., $n_i^B= \bar{n}^B, i\in\mathcal{M}, n_i^H= \bar{n}^H, i\in\mathcal{M}$, then $x_i = m \bar{n}^B$, and $x_i = m \bar{n}^H$ 
%
\subsubsection{Activity Level of Bots}
The rate $\bar{p}_i$ in (\ref{eq:r_iH}) indicates the activity level of bots when they respond to or poll information from  C\&C node $i$. This level of activity is often correlated with parameter $r$, the rate of sending out spamming messages to the social network. Assume that all C\&C bots are assumed to be identical, i.e., $\bar{p}_i=\bar{p}, i\in\mathcal{M},$ then we can let $\bar{p} = \alpha r$, where $\bar{p}$ is in messages/sec, $r$ is in messages/sec and $\alpha \in\mathbb{R}_{++}$ is a unitless positive parameter.


\begin{figure*}[t]
\centering
$\begin{array}{ccc}
\includegraphics[width=2in]{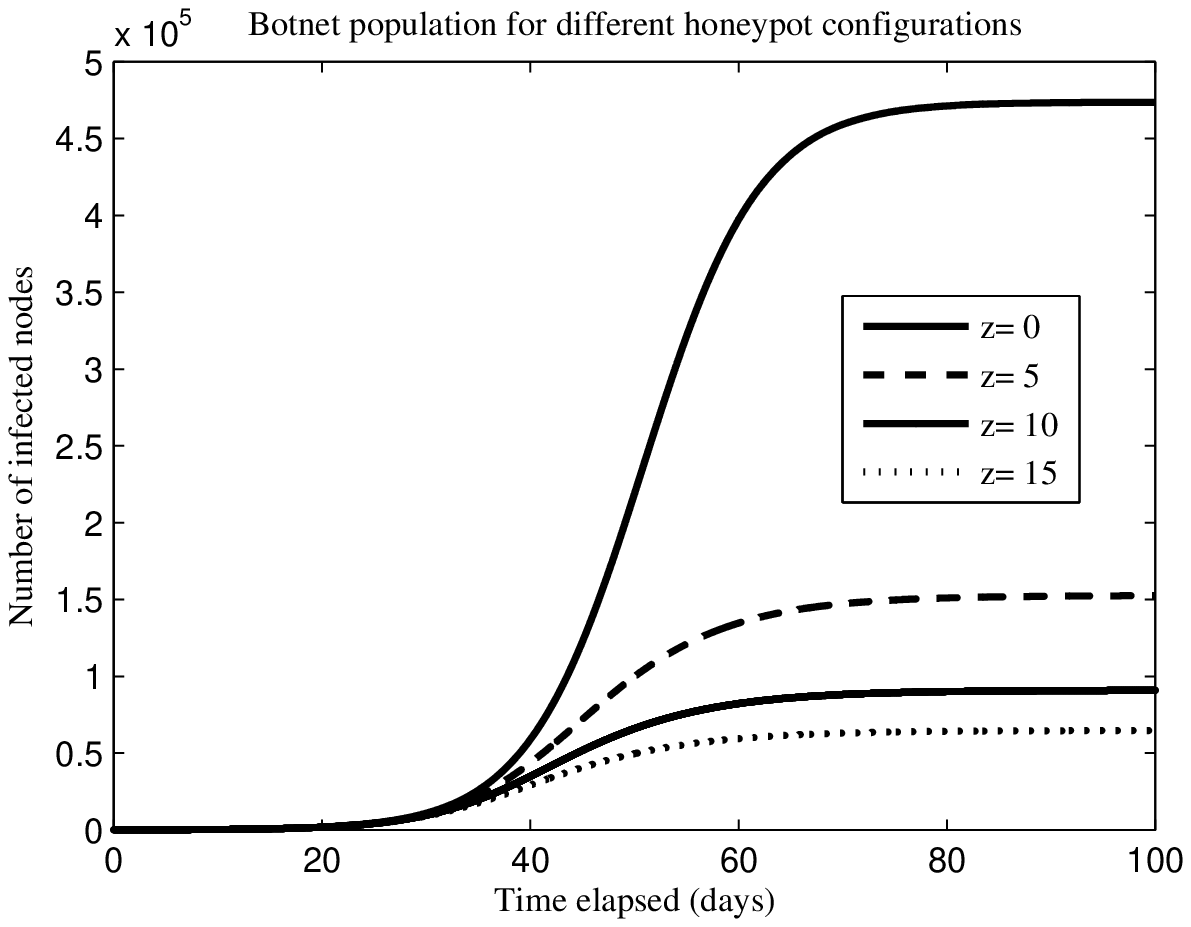} &
\includegraphics[width=2in]{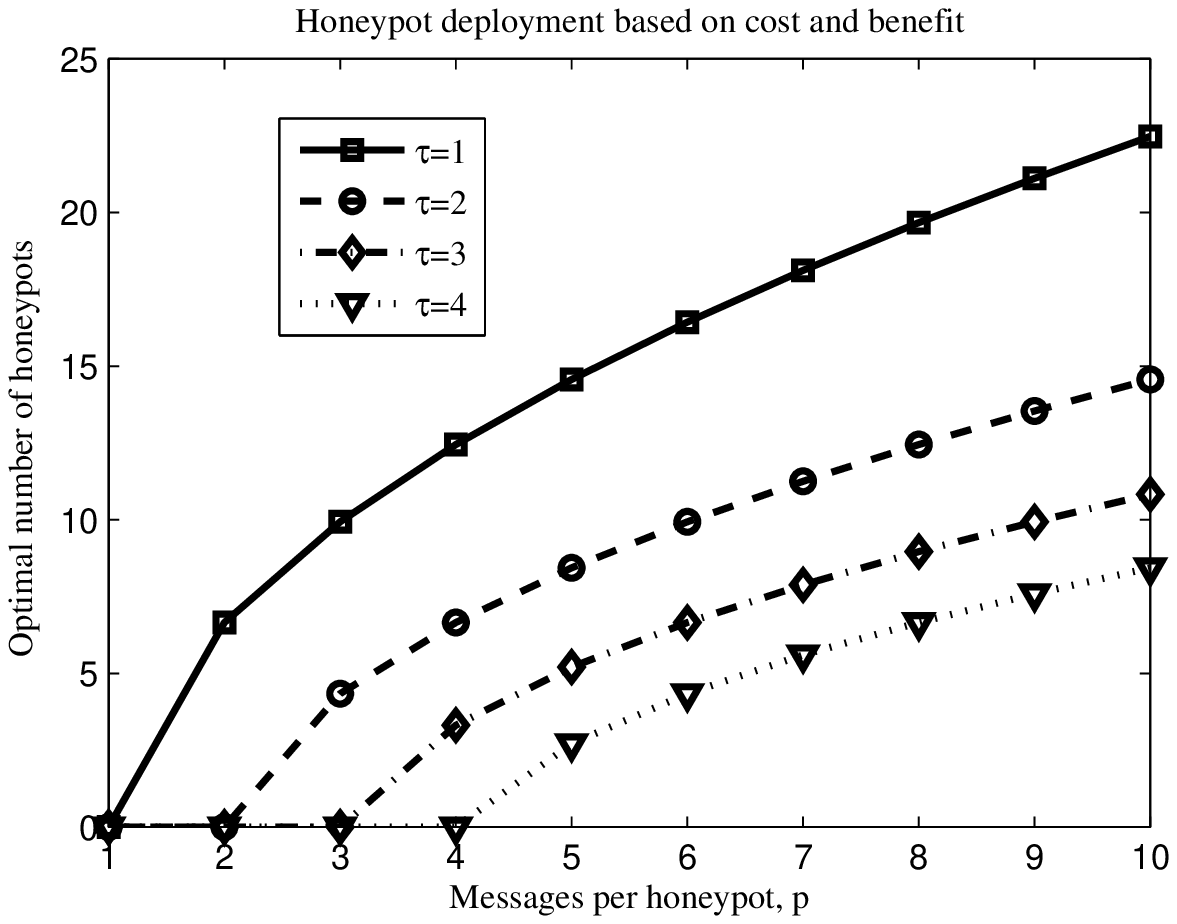} &
\includegraphics[width=2in]{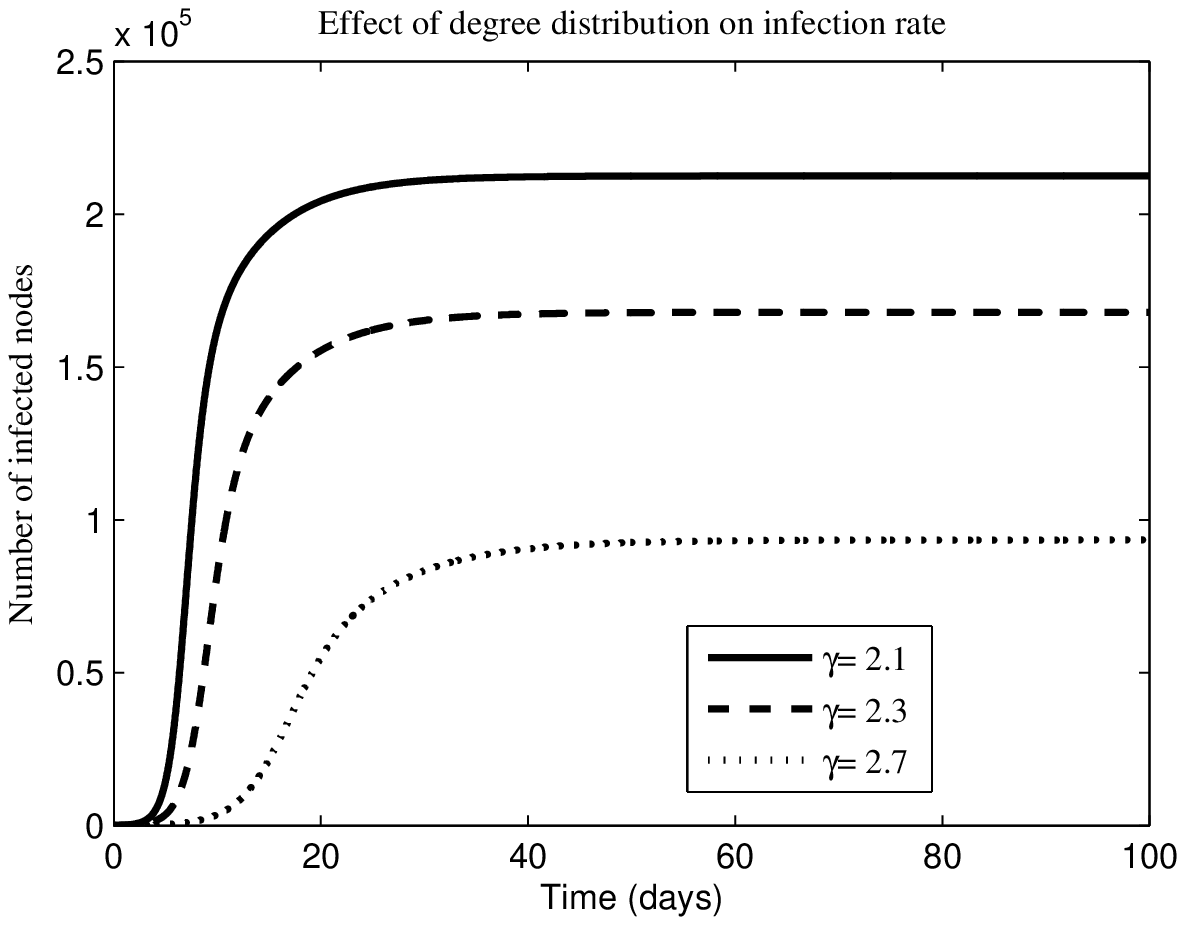} \\
(a) & (b) & (c)
\end{array}$
\caption{Simulation of our framework for a network of $N=10^{6}$ users, where each user has probability $q=0.01$ of following a malicious link, message are sent at a rate of $0.4$ messages per bot per day, and infected nodes are cleaned after $5$ days on average. (a) Effect of increasing the number of honeypot nodes on the botnet population.  Deployment of a small number of honeypots can greatly reduce the number of bots present.  Note that the population converges quickly to its equilibrium value. (b) The optimum number of bots based on (\ref{eq:z_opt}) for different costs $\tau$ and benefits $p$.  The total number of honeypots remains small for each case.  (c) Effect of degree distribution on the botnet population for number of honeypots $z=5$.  Each network is scale-free, with exponent $\gamma$ varying between networks.  A higher connectivity results in a larger number of bots.}
\label{fig:simulation}
\end{figure*}

\subsection{Cross-Layer Optimal Honeybot Deployment}
In what follows, we first derive the optimal honeybot deployment when the benefit from each honeypot is measurable.  We then combine the analysis of Sections \ref{sec:exploitation} and \ref{sec:deployment} to determine the optimal honeybot deployment, taking into account the behavior of the deployed nodes during the exploitation phase.

The goal of the honeypot operator is to maximize the number of blacklisted links that are reported to the social network.  Based on the analysis of Corollary \ref{cor3}, we assume that the number of blacklisted links is proportional to the number of honeypot nodes in the botnet in steady-state, $\mathbf{x}_{2}^{\ast}$.  The variable is the number of honeypot nodes that have not yet been inducted into the botnet, $z$.  This leads to a utility function given by $V_{H}(z) = px_{2}^{\ast}(z) - \tau(x_{2}^{\ast} + z)$,
where $p$ and $\tau$ represent the benefit (information gathered) and cost of maintaining a single honeypot node.  Substituting (\ref{eq:dynamic_equilibria}) yields
\begin{IEEEeqnarray}{rCl}
\label{eq:utility_honeypot}
\nonumber V_{H}(z) &=& p\frac{\left(rd - \frac{\mu_{1}}{q}\right)z}{\frac{rdkz}{M} + \mu_{2}} - \tau\left(\frac{\left(rd - \frac{\mu_{1}}{q}\right)z}{\frac{rdkz}{M} + \mu_{2}} + z\right)\\
 &=& \frac{(p - \tau)\left(rd - \frac{\mu_{1}}{q}\right)z}{\frac{rdkz}{M} + \mu_{2}} - \tau z.
\end{IEEEeqnarray}
The value of $z$ that maximizes (\ref{eq:utility_honeypot}) is given by the following proposition.
\begin{proposition}
\label{prop:optimum_z}
The optimum value of $z$ is given by
\begin{equation}
\label{eq:z_opt}
z^{\ast} = M\left(\frac{-\mu_{2} + \sqrt{(p-\tau)(rd - \frac{\mu_{1}}{q})\mu_{2}/\tau}}{rdk}\right).
\end{equation}
\end{proposition}

\begin{proof}
Differentiating $V_{H}(z)$ with respect to $z$ yields
\begin{displaymath}
\frac{\partial V_{H}}{\partial z} = \frac{(p-\tau)\left(rd - \frac{\mu_{1}}{q}\right)\mu_{2}}{\left(\frac{rdk}{M}z + \mu_{2}\right)^{2}} - \tau.
\end{displaymath}
By inspection, $\frac{\partial V_{H}}{\partial z}$ is a strictly decreasing function of $z$, so that $V_{H}(z)$ is strictly concave.
Setting this expression equal to zero implies (\ref{eq:z_opt}).
\end{proof}

\begin{remark}
\label{remark:z_opt}
Eq. (\ref{eq:z_opt}) has several implications for the design of honeybot systems.  First, for malware that propagates rapidly (corresponding to a large $rd$ value), fewer honeybots are needed, since the malware will quickly spread to the deployed honeybot.  Second, if $\mu_{2}$ is large, then honeybots are rapidly detected and removed by the botmaster, and hence the cost of deploying honeybots outweighs the benefits.
\end{remark}

The utility function (\ref{eq:utility_honeypot}) can be augmented by incorporating the impact on the exploitation phase.  In particular, (\ref{eq:r_iH}) implies that $p = \frac{1}{1 + \frac{\beta_{i}^{H}x_{1}^{\ast}\overline{T}_{i}r}{T_{iH}}}$, which we write as
$p = \frac{1}{1+ \zeta x_{1}^{\ast}} \approx \frac{1}{\zeta x_{1}^{\ast}}$
when the number of bots is sufficiently large.  The utility function $V_{H}$ can then be written as
\begin{equation}
\label{eq:combined_utility_function}
V_{H} = \left(\frac{1}{\zeta x_{1}^{\ast}} - \tau\right)x_{2}^{\ast} - \tau z
\end{equation}
An efficient algorithm for maximizing (\ref{eq:combined_utility_function}) can be derived using the following theorem.
\begin{theorem}
\label{theorem:combined_optimization}
The problem of selecting $z$ to maximize $V_{H}$ in (\ref{eq:combined_utility_function}) is equivalent to
\begin{IEEEeqnarray}{ll}
\nonumber\max_{\theta, \phi, x_{2}^{\ast}, z} & \frac{1}{\zeta}\left(-\frac{rdq\theta^{2}}{Nrdq\phi - \mu_{1}} + \frac{M/4k}{N(rdq\left(1 - \frac{kx_{2}^{\ast}}{M}\right) - \mu_{1})}\right)  \\
 & -\tau x_{2}^{\ast} - \tau z, \\
\label{eq:constraint1}
\nonumber \mbox{s.t.} & \theta = x_{2}^{\ast} - \frac{M}{2k}, \\
\label{eq:constraint2}
 & \phi = 1- \frac{kx_{2}^{\ast}}{M}, \\
 \label{eq:constraint3}
& x_{2}^{\ast} \leq \frac{\left(rd - \frac{\mu_{1}}{q}\right)z}{\frac{rdkz}{M} + \mu_{2}}, \\
\label{eq:constraint4}
\nonumber & \frac{1}{\zeta}\frac{rdq\mu_{2}M + rdkz\mu_{1}}{(rdq-\mu_{1})\mu_{2}M} \geq \tau ,\\
 \label{eq:constraint5}
\nonumber &  z \geq 0, \ \ 0 \leq x_{2}^{\ast} \leq N,
 \end{IEEEeqnarray}
 which is a convex program.
 \end{theorem}

 \begin{proof}
The optimization problem of selecting $z$ to maximize $V_{H}$ can be written as
\begin{equation}
\begin{array}{cc}
\mbox{maximize} & \left(\frac{1}{\zeta x_{1}^{\ast}(z)} - \tau\right)x_{2}^{\ast}(z) - \tau z \\
z\in\mathbb{R}_+ & 
\end{array}
\end{equation}
If $\frac{1}{\zeta x_{1}^{\ast}(z)} - \tau < 0$, then the objective function is monotone decreasing in $z$, leading to an optimal value $z = 0$.  To avoid this, we require $\frac{1}{\zeta x_{1}^{\ast}(z)} \geq \tau$, leading to constraint (\ref{eq:constraint4}).
Using (\ref{eq:dynamic_equilibria}), we have
\begin{equation}
\frac{x_{2}^{\ast}}{\zeta x_{1}^{\ast}} = \frac{\left(x_{2}^{\ast} - \frac{kx_{2}^{\ast 2}}{M}\right)rdq}{N\left(rdq\left(1 - \frac{kx_{2}^{\ast}}{M}\right)-\mu_{1}\right)} = \frac{rdq\left(-\frac{k}{M}\left(x_{2}^{\ast} - \frac{M}{2k}\right)^{2} + \frac{M}{4k}\right)}{N\left(rdq\left(1 - \frac{kx_{2}^{\ast}}{M}\right) - \mu_{1}\right)}.
\end{equation}
Substituting $\theta = \left(x_{2}^{\ast} - \frac{M}{2k}\right)$ and $\phi = \left(1 - \frac{kx_{2}^{\ast}}{M}\right)$ leads to the objective function
\begin{IEEEeqnarray}{rCl}
\label{eq:equivalent_objective}
\nonumber V_{H}(x_{2}^{\ast}, z, \theta, \phi) &=& \frac{1}{\zeta}\left(-\frac{rdq\theta^{2}}{Nrdq\phi - \mu_{1}}+\frac{M/4k}{N(rdq\left(1 - \frac{kx_{2}^{\ast}}{M}\right) - \mu_{1})}\right) \\
 && - \tau x_{2}^{\ast} - \tau z
\end{IEEEeqnarray}
Since quadratic over linear and inverse functions are convex, the first two terms of (\ref{eq:equivalent_objective}) are concave, and hence $V_{H}$ is concave.

Finally, the fact that the objective function is increasing as a function of $x_{2}^{\ast}$ and decreasing as a function of $z$ implies that the constraint (\ref{eq:constraint3}) holds with equality at the optimum, so that the relationship between $x_{2}^{\ast}$ and $z$ in (\ref{eq:dynamic_equilibria}) is satisfied.  This constraint is convex due to Proposition \ref{prop:optimum_z}.
\end{proof}

The convex optimization approach presented in Theorem \ref{theorem:combined_optimization} is used to select a honeybot deployment strategy in order to maximize the level of infiltration into the botnet and the amount of data gathered during the exploitation phase.  Once inducted into the botnet, the honeybots follow the Stackelberg equilibrium strategy of Section \ref{sec:exploitation} and use the collected data to generate malware signatures and create URL blacklists.  The parameters of (\ref{eq:combined_utility_function}) are updated in response to changes in botnet behavior observed during the exploitation phase.


\section{Simulation Study}
\label{sec:simulation}
We evaluated our proposed method using Matlab simulation study, described as follows.  A network consisting of $N=10^{6}$ nodes was generated, with degree $d=100$ (consistent with observations of the average degree of social networks~\cite{gjoka2011crawling}).  The rate at which malware messages are sent is given by $r=0.4$ messages per bot day, and the rate at which nodes are disinfected and removed from the botnet is $\mu_{1} = 0.2$, an average lifetime for each bot of $5$ days.  These statistics are based on the empirical observations of~\cite{Thomas2010}.  Based on~\cite{grier2010spam}, we estimate that the probability of a user clicking on a spam link is given by $q = 0.01$.  It is assumed that the fraction of malware links given to each bot is equal to $k/M = 0.01$.  The rate at which honeybots are detected and removed is equal to $\mu_{2} = 0.5$.  In each case, we assume that there are $50$ infected nodes and $0$ honeybots present in the network initially.

The population dynamics of the bots, described by (\ref{eq:bot_dynamics}) and (\ref{eq:honeypot_dynamics}), are shown in Fig. \ref{fig:simulation}(a).  Each curve represents the number of infected users over time for a different level of honeybot activity, as described by the parameter $z$.  In each case, the number of bots converges to its equilibrium value.  The top curve (solid line) assumes $z=0$, i.e. no deception takes place and malicious links are detected through blacklists only.  Employing deception through honeybots significantly reduces the botnet population, even when the number of honeybots is small relative to the population size.  As additional honeybots are added, the botnet population continues to decline.  However, the marginal benefit of adding a honeybot decreases as the number of honeybots grows large.  

The optimum number of honeybots depends on the cost of introducing and maintaining honeybots, denoted $\tau$, as well as the benefit $p$ from each honeybot, as described in (\ref{eq:z_opt}).  The optimum number of honeybots is given in Fig. \ref{fig:simulation}(b).  As the cost of introducing new honeybots is reduced, the optimal number of honeybots increases.  In each case, the optimum number of honeybots remains small, at around $25$ nodes, relative to the total network population of $10^{6}$ nodes.

The effect of a heterogeneous degree distribution is shown in Fig. \ref{fig:simulation}(c).  The degree distribution was chosen to be scale-free, so that the probability that a node has degree $d$ was proportional to $d^{-\gamma}$.  Hence a higher value of $\gamma$ corresponds to a less-connected network.  The parameter $\gamma$ had a significant impact on the rate of propagation of the botnet, even through for the chosen values of $\gamma$ the average degrees of the three networks were similar.

\section{Conclusion}
\label{sec:conclusion}
In this paper, we studied the problem of defending against social botnet attacks through deception.  We considered a defense mechanism in which fake honeybot accounts are deployed and infiltrate the botnet, impersonating infected users.  The infiltrating honeybots gather information from command and control messages, which are used to form malware signatures or add spam links to URL blacklists.  We
introduced a framework for SOcial network Deception and EXploitation through hOneybots (SODEXO), which provides an analytical approach to modeling and designing social honeybot defenses.  We decomposed SODEXO into deployment and exploitation components.

In the deployment component, we model the population dynamics of the infected users and honeybots, and show how the infected population is affected by the number of honeybots introduced.  We derive the steady-state populations of infected users and honeybots and prove the stability of the equilibrium point.  In the exploitation component, we formulate a Stackelberg game between the botmaster and the honeybots and determine the amount of information gathered by the honeybot in equilibrium.  The two components are combined in the Protection and Alert System (PAS), which chooses an optimal deployment strategy based on the observed behavior of the botnet and the information gathered by the honeybots.  Our results are supported by simulation studies, which show that a small number of honeybots  significantly decrease the infected population of a large social network.

\bibliographystyle{ieeetr}

\bibliography{Zhu_SODEXO}

\end{document}